\documentclass[conference]{IEEEtran}
\usepackage{mathrsfs}
\usepackage{psfrag}
\usepackage{graphicx}
\usepackage{fancyhdr,epsfig}
\usepackage{mathtools}
\usepackage{research5IEEE}
\usepackage{amsmath}
\usepackage{amsfonts}
\usepackage{amssymb}
\usepackage{amsthm}
\usepackage{array}

\newtheorem{Theorem}{Theorem}

\newtheorem{Remark}{Remark}
\newtheorem{Example}{Example}
\newtheorem{Corollary}{Corollary}
\newtheorem{Definition}{Definition}

\begin{document}
\title{ Any Positive Feedback Rate  Increases the Capacity of Strictly Less-Noisy Broadcast Channels}
\author{
  \IEEEauthorblockN{Youlong Wu}
  \IEEEauthorblockA{Comm. and Electr. Department\\
    Telecom ParisTech\\
    Paris, France\\
     youlong.wu@telecom-paristech.fr }
  \and
  \IEEEauthorblockN{Mich\`{e}le Wigger}
  \IEEEauthorblockA{Comm. and Electr. Department\\
    Telecom ParisTech\\
    Paris, France\\
    michele.wigger@telecom-paristech.fr }
}
\maketitle

\begin{abstract}
We propose two coding schemes for discrete memoryless broadcast channels (DMBCs) with rate-limited feedback from only one receiver. For any positive feedback rate and for the class of strictly less-noisy DMBCs, our schemes strictly improve over the no-feedback capacity region.
\end{abstract}


\section{Introduction}
We study the capacity region of discrete memoryless broadcast channels (DMBCs) with feedback. It is known that for \emph{physically degraded} DMBCs, feedback does not change the capacity region \cite{gamal'78}. In contrast, there exist a few specific examples of not physically degraded DMBCs where feedback strictly enlarges the capacity region \cite{kramer,dueck,wigger}. For general DMBCs with feedback, achievable regions have been proposed in \cite{kramer, wigger,venkataramananpradhan11}. But due to their complexity it is hard to evaluate these regions or to obtain general insights from them.

The usefulness of feedback has also been shown for memoryless {Gaussian} broadcast channels (BCs) \cite{ozarow,bhaskaran}. Achievable regions and the  asymptotic high-SNR sum-capacity of Gaussian BCs have been presented in \cite{minero,wigger'10}.

In this paper we propose two coding schemes for general DMBCs with feedback which lead to relatively simple achievable regions (with only two auxiliary random variables). In our schemes it suffices that there is a \emph{rate-limited} feedback link from the weaker receiver and no feedback link from the stronger receiver.

For the class of \emph{strictly less-noisy} DMBCs (see Definition~1 ahead), our schemes strictly improve over the no-feedback capacity region for any positive feedback rate, no matter how small. As we will see, the class of strictly less-noisy DMBCs includes, for example, all asymmetric {binary symmetric BCs} (BS-BCs) and all asymmetric {binary erasure BCs} (BE-BCs).


We conclude this section with some notation. We write $Z\sim \textnormal{Bern}(p)$ to indicate that $Z$ is  Bernoulli-$p$, and we use the definitions $\bar{a}:=(1-a)$ and $a*b:=\bar{a}b+a\bar{b}$, for  $a, b\in[0,1]$. Also, for any positive integer $k$, $A^k$ stands for the $k$-tuple $A_1,\ldots, A_k$.
Given a set $\set{S}\in \Reals^2$, we denote by $\textnormal{bd}(\set{S})$ and $\textnormal{int}(\set{S})$  the boundary and the interior of $\set{S}$.

\section{Channel model}\label{sec:system}

\begin{figure}[!t]
\centering
\includegraphics[width=0.45\textwidth]{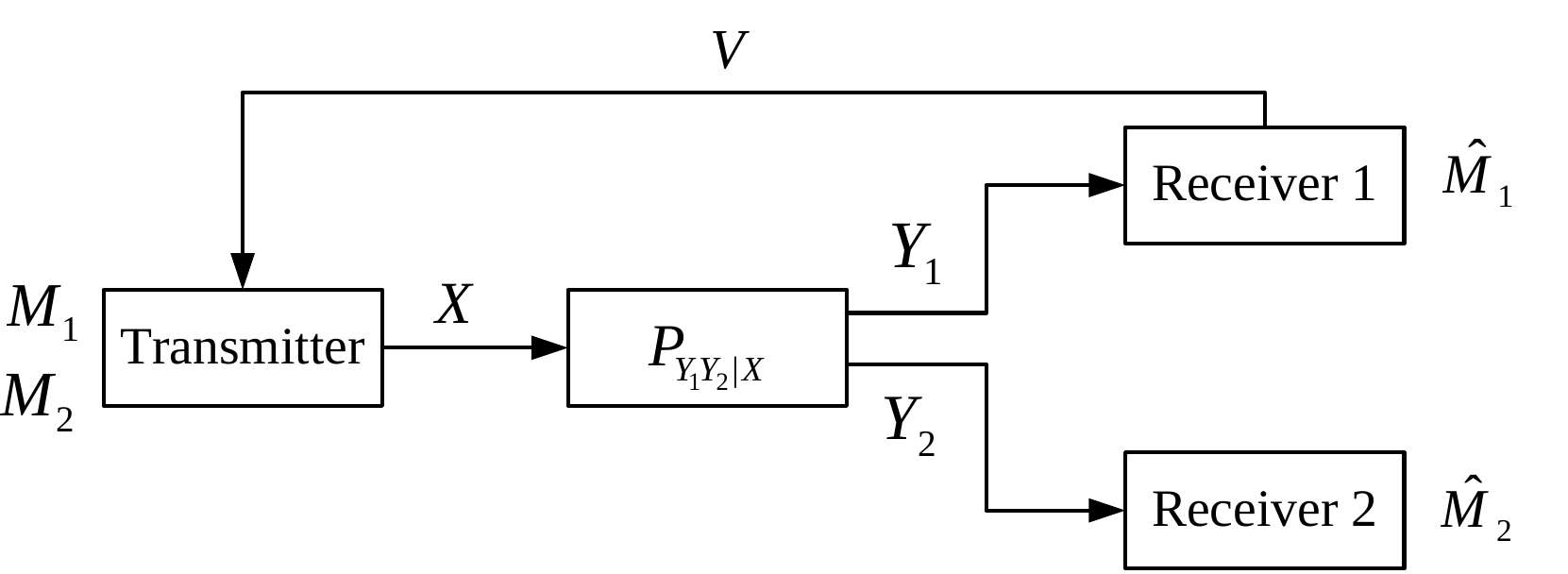}
\caption{Broadcast channel with rate-limited feedback from Receiver 1.} \label{fig:BCFBmodel}
\vspace{-5mm}
\end{figure}

Communication takes place over a DMBC with rate-limited feedback from Receiver 1, see Figure~\ref{fig:BCFBmodel}. The setup is characterized by the finite input alphabet $\set{X}$, the finite output alphabets $\set{Y}_1$ and $\set{Y}_2$, the channel law $P_{Y_1Y_2|X}$, and the nonnegative feedback rate $R_\textnormal{FB}$.
Specifically, if at  discrete-time $t$ the transmitter sends the channel input $x_t\in \set{X}$, then Receiver~$i\in\{1,2\}$ observes the output $Y_{i,t}\in \set{Y}_i$, where the pair $(Y_{1,t}, Y_{2,t})\sim P_{Y_1Y_2|X}(\cdot, \cdot|x_t)$. Also, after observing $Y_{1,t}$, Receiver~1 can send a feedback signal $V_{t}\in \set{V}_t$ to the transmitter, where $\set{V}_t$ denotes the finite alphabet of $V_t$ and is a design parameter of a scheme. The feedback link is assumed to be instantaneous and noiseless---i.e., the transmitter observes $V_{t}$ before it has to produce the next input $X_{t+1}$---but \emph{rate-limited} to $R_{\textnormal{FB}}$ bits on average. Thus,
if the transmission takes place over a total blocklength $N$, then
\begin{equation} \label{consFB0}
|\set{V}_1|\cdots|\set{V}_N|\leq 2^{NR_\textnormal{FB}}.
\end{equation}

The goal of the communication is that the transmitter conveys two independent private messages $M_1\in \{1,\ldots,\lfloor 2^{NR_1} \rfloor\}$ and $M_2\in \{1,\ldots,\lfloor 2^{NR_2} \rfloor\}$, to Receiver 1 and 2, respectively. Each $M_i$, $i=1,2$, is uniformly distributed over the set $\set{M}_i:=\{1,\ldots,\lfloor 2^{NR_i} \rfloor\}$, where $R_i$ denotes the private rate of transmission of Receiver $i$.

The transmitter is comprised of a sequence of encoding functions $\big\{f^{(N)}_t\big\}^N_{t=1}$ of the form
$ f_t^{(N)}: \set{M}_1\times\set{M}_2\times{\set{V}}_1\times\cdots \times{\set{V}}_{t-1}\to \mathbb{R}$
that is used to produce the channel inputs as
\begin{equation}
X_t=f^{(N)}_t\big(M_1,M_2,{V}_1,\ldots,{V}_{t-1}\big), \qquad t\in\{1,\ldots, N\}.
\end{equation}

Receiver~1 is comprised of a sequence of feedback-encoding functions $\{\psi^{(N)}_t\}_{t=1}^N$ of the form
$
\psi^{(N)}_t: \mathbb{R}^{t}\to \set{V}_t
$
that is used to produce the symbols
\begin{equation}
V_t=\psi^{(N)}_t(Y_{1,1},\ldots,Y_{1,t}), \qquad t\in\{1,\ldots, N\},
\end{equation}
sent over the feedback link,
and of a decoding function
$
\Phi^{(N)}_1: \mathbb{R}^N\to\set{M}_1$
used to produce a guess of Message $M_1$:
\begin{equation}
\hat{M}_1=\Phi^{(N)}_1(Y_1^N).
\end{equation}

Receiver~2 is only comprised of a decoding function $\Phi^{(N)}_2: \mathbb{R}^N\to\set{M}_2$
used to produce a guess of Message $M_2$:
\begin{equation}
\hat{M}_2=\Phi^{(N)}_2(Y_2^N).
\end{equation}

A rate region $(R_1, R_2)$ is called achievable if for every blocklength $N$, there exists a set of $N$ encoding functions $\big\{f_t^{(N)}\big\}_{t=1}^N$, two decoding functions $\Phi^{(N)}_1$ and $\Phi^{(N)}_2$, $N$ feedback alphabets $\set{V}_1,\ldots, \set{V}_N$ satisfying~\eqref{consFB0}, and $N$ feedback-encoding functions $\big\{\psi_{t}^{(N)} \big\}_{t=1}^N$  such that the error probability
\begin{equation}\label{errorprob}\Pr(M_1\neq \hat{M}_1~\textnormal{or}~M_2 \neq \hat{M}_2)\end{equation}
tends to zero as the blocklength $N$ tends to infinity. The closure of the set of achievable rate pairs $(R_1, R_2)$ is called the \textit{feedback capacity region} and is denoted by $\set{C}_\textnormal{FB}(R_\textnormal{FB})$.

In the special case $R_\textnormal{FB}=0$ the feedback signals are constant and the setup is equivalent to a setup without feedback. We denote the capacity region for this setup by $\set{C}_\textnormal{NoFB}$. 

We are particularly interested in the following DMBCs:
\begin{Definition}
  A DMBC is called \emph{less-noisy} \cite{gamal'79} if
 \begin{equation}\label{LN}
  I(U;Y_2)\geq I(U;Y_1)
  \end{equation}
holds for all probability mass functions (pmf) $P_{UX}P_{Y_1Y_2|X}$. We call a DMBC \emph{strictly} less-noisy if (\ref{LN}) holds with strict inequality whenever $I(U;Y_1)>0$.
\end{Definition}
In general, the capacity region of DMBCs with and without feedback are unknown. For the class of less-noisy  BCs, the no-feedback capacity region $\C_\textnormal{NoFB}$ is known \cite{gamal'79}. It is the set of all nonnegative rate pairs $(R_1,R_2)$ that satisfy
\begin{subequations}\label{NoFB}
\begin{IEEEeqnarray}{rCl}
R_1&\leq& I(U;Y_1) \\
R_2&\leq& I(X;Y_2|U),
\end{IEEEeqnarray}
\end{subequations}
for some pmf $P_{UX}$, where the cardinality of the auxiliary random variable $U$ satisfies $|\set{U}|\leq \min\{|\set{X}|,|\set{Y}_1|,|\set{Y}_2|\}+1$.

We will also need these definitions. A BC is said \emph{physically degraded} if $X- Y_2- Y_1$ forms a Markov chain. For physically degraded BCs the capacity regions with and without feedback are the same and given by the constraints in \eqref{NoFB} \cite{gamal'78}.

For comparison, we introduce the notion of \emph{enhanced DMBC}, which is obtained from the original DMBC by revealing outputs $Y_1^n$ to Receiver 2. The enhanced DMBC is {physically degraded} and thus, with and without feedback, its capacity region is described by~\eqref{NoFB} where $Y_2$ needs to be replaced by $(Y_1,Y_2)$. We denote this capacity region by $\set{C}_\textnormal{Enh}$.

\section{Main Results}\label{sec:results}

\begin{Theorem}\label{theo1}
For less-noisy DMBCs, the capacity region $\set{C}_\textnormal{FB}(R_\textnormal{FB})$ includes the set $\set{R}_\textnormal{in,1}$ of all nonnegative rate pairs $(R_1,R_2)$ that satisfy
\begin{subequations}\label{rate1}
\begin{IEEEeqnarray}{rCl}
R_1&\leq& I(U;Y_1|Q) \label{1a}\\
R_1&\leq&I(U;Y_2|Q)-I(\tilde{Y};Y_1|UY_2Q) \label{1b}\\
R_2&\leq& I(X;\tilde{Y}Y_2|UQ) \label{1c}
\end{IEEEeqnarray}
\end{subequations}
for some pmf $P_QP_{U|Q}P_{X|UQ}P_{Y_1Y_2|X}P_{\tilde{Y}|UY_1Q}$ satisfying
\begin{equation}\label{fbQ1}
I(\tilde{Y};Y_1|UY_2Q)\leq {R}_\textnormal{FB}.
\end{equation}
\end{Theorem}

\begin{proof}
See Section \ref{sec:sche1}.
\end{proof}
\begin{Corollary}\label{rateCoro}
For less-noisy DMBCs, the capacity region $\set{C}_\textnormal{FB}(R_\textnormal{FB})$ includes the set of all nonnegative rate pairs $(R_1,R_2)$ that satisfy
\begin{subequations}
\begin{IEEEeqnarray}{rCl}
R_1&\leq& I(U;Y_1|Q) \label{1aQ}\\
R_2&\leq& I(X;\tilde{Y}Y_2|UQ) \label{1cQ}
\end{IEEEeqnarray}
\end{subequations}
for some pmf $P_QP_{U|Q}P_{X|UQ}P_{Y_1Y_2|X}P_{\tilde{Y}|UY_1Q}$ satisfying
\begin{equation}\label{fbQ2}
I(\tilde{Y};Y_1|UY_2Q) \leq\min\{R_\textnormal{FB}, I(U;Y_2|Q)-I(U;Y_1|Q)\}.
\end{equation}
\end{Corollary}
\begin{Example}
 Consider asymmetric BS-BCs, where $P_{Y_1Y_2|X}$ is described by
   \begin{IEEEeqnarray}{rCl}
   Y_i&=&X\oplus Z_i, \quad i\in{1,2}
 \end{IEEEeqnarray}
for $Z_1\sim \textnormal{Bern}(p_1)$ and $Z_2\sim\textnormal{Bern}(p_2)$ independent of each other with $0<p_2<p_1<1/2$. We evaluate the region $\set{R}_\textnormal{in,1}$ in Theorem~\ref{theo1}  for distributions of the form
\begin{subequations}
   \begin{IEEEeqnarray}{rCl}
   X&=& U\oplus W_1\\
    \tilde{Y}&=& U\oplus Y_1\oplus W_2
\end{IEEEeqnarray}
 \end{subequations}
with $U\sim\textnormal{Bern}(1/2)$, $W_1\sim\textnormal{Bern}(\alpha)$, and $W_2\sim\textnormal{Bern}(\beta)$ independent of each other and of the pair $(Z_1,Z_2)$ and $\alpha, \beta\in[0,1/2]$. This results in the region of all nonnegative rate pairs $(R_1,R_2)$ that satisfy
 \begin{subequations}\label{eq:regBS}
   \begin{IEEEeqnarray}{rCl}
   R_1&\leq&  1-H(\alpha *p_1)\\
   R_1&\leq& 1+H(\beta)-H(\alpha_1,\alpha_2,\alpha_3,\alpha_4)\nonumber\\
   R_2&\leq&  H(\alpha_1,\alpha_2,\alpha_3,\alpha_4)-H(p_2)-H(p_1*\beta) 
 \end{IEEEeqnarray}
 \end{subequations}
for some $\alpha, \beta\in[0,1/2]$ satisfying
 \[H(\alpha_1,\alpha_2,\alpha_3,\alpha_4)-H(\alpha*{p_2})-H(\beta)\leq R_\textnormal{FB}\]
 where
\begin{IEEEeqnarray*}{rCl}
 \alpha_1&=&(p_1*\beta)p_2\alpha +(1-p_1*\beta)\bar{p_2}\bar{\alpha}\nonumber\\
 \alpha_2&=&(p_1*\beta)\bar{p_2}\alpha+(1-p_1*\beta)p_2\bar{\alpha}\nonumber\\
 \alpha_3&=&(p_1*\beta)\bar{p_2}\bar{\alpha}+(1-p_1*\beta)p_2\alpha\nonumber\\
 \alpha_4&=&(p_1*\beta)p_2\bar{\alpha}+(1-p_1*\beta)\bar{p_2}\alpha .
  \end{IEEEeqnarray*}
Figure \ref{fig:rateBSBC} compares this region to $\set{C}_\textnormal{NoFB}$ when $p_2=0.1$, $p_1\in\{0.2,0.25,0.3\}$, and $R_\textnormal{FB}=0.85$.
\end{Example}

 \begin{figure}[!t]
\centering
\includegraphics[width=0.5\textwidth]{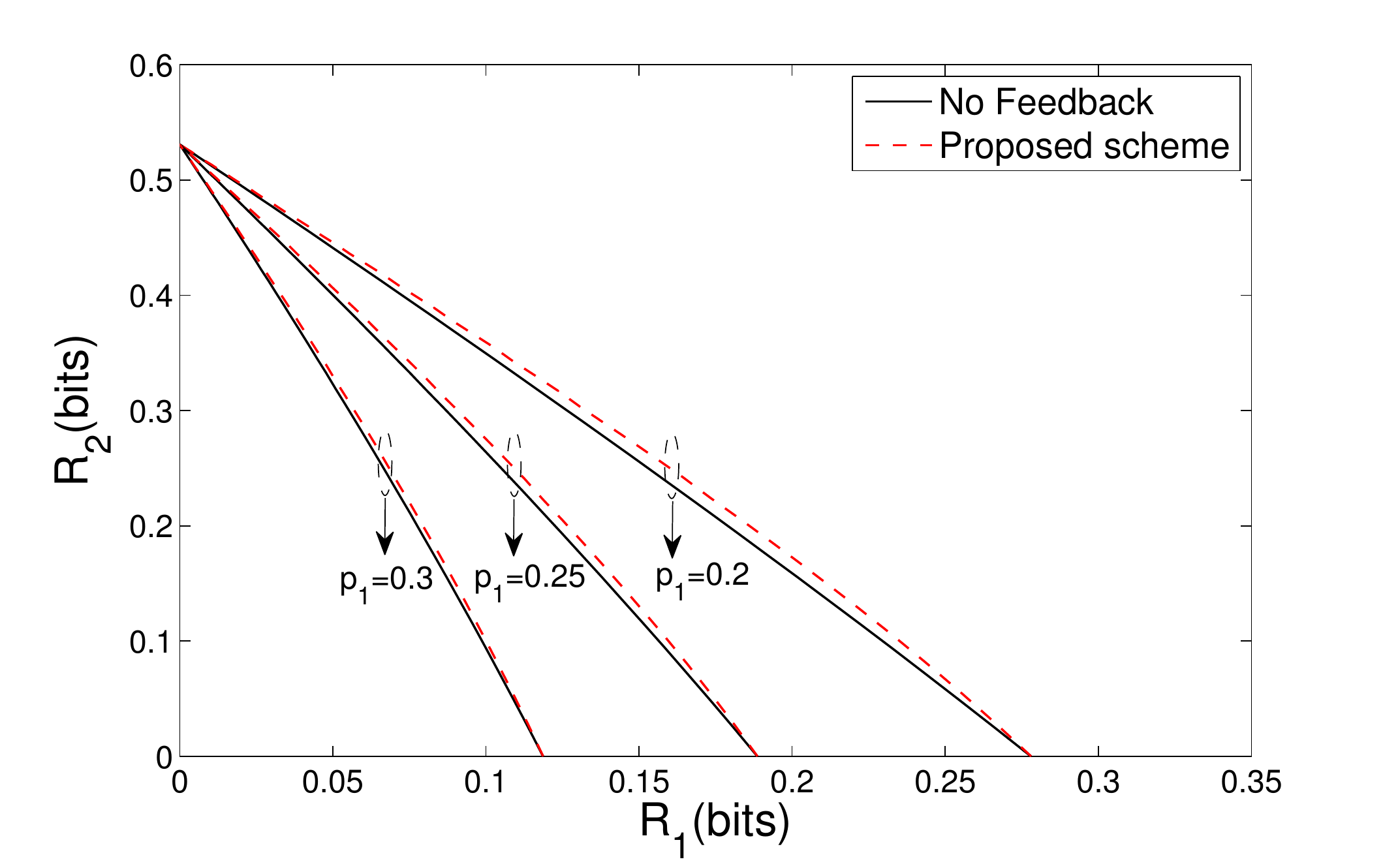}
\caption{$\set{C}_{\textnormal{NoFB}}$  and the achievable region in~\eqref{eq:regBS} are plotted for BS-BCs with parameters $p_2=0.1$ and $p_1\in\{0.2,0.25,0.3\}$ and feedback rate $R_\textnormal{FB}=0.85$.} \label{fig:rateBSBC}
\vspace{-4mm}
\end{figure}

\begin{Theorem}\label{theo2}
For  less-noisy DMBCs, the capacity region $\set{C}_{\textnormal{FB}}(R_{\textnormal{FB}})$ includes the set $\set{R}_\textnormal{in,2}$ of all nonnegative rate pairs $(R_1,R_2)$ that satisfy
\begin{subequations}
\begin{IEEEeqnarray}{rCl}
{R}_1&\leq&I(U;Y_1|Q)\label{2a}\\
{R}_1&\leq&I(X;Y_2|Q)+I(X;\tilde{Y}|UY_2Q)\nonumber \\ & & -I(\tilde{Y};Y_1|UY_2Q)\label{2b}\\
{R}_2&\leq&I(X;Y_2\tilde{Y}|UQ)\label{2c}\\
R_1+{R}_2&\leq&I(X;Y_2|Q)+I(XY_2;\tilde{Y}|UQ)\nonumber\\
&~&~+I(X;\tilde{Y}|UY_2Q)-I(\tilde{Y};Y_1|UY_2Q) \label{2d}\IEEEeqnarraynumspace
\end{IEEEeqnarray}
\end{subequations}
for some pmf $P_QP_{U|Q}P_{X|UQ}P_{Y_1Y_2|X}P_{\tilde{Y}|UY_1Q}$ satisfying
\begin{equation} \label{fbQ2inner2}
I(\tilde{Y};Y_1|UY_2Q)-I(X;\tilde{Y}|UY_2Q)\leq {R}_\textnormal{FB}.
\end{equation}
\end{Theorem}
\begin{proof}
The  scheme achieving $\set{R}_{\textnormal{in,2}}$ is similar to the  scheme achieving $\mathcal{R}_\textnormal{in,1}$ described in Section~\ref{sec:sche1}, but  Receiver~2 applies backward decoding as opposed to sliding-window decoding. Details and analysis are omitted.
\end{proof}

\begin{Remark} \label{remark1}
$\set{R}_\textnormal{in,1} \subseteq \set{R}_\textnormal{in,2}$.
\end{Remark}
This holds because Constraints (\ref{1a}) and (\ref{1c}) are equivalent to Constraints (\ref{2a}) and (\ref{2c}), respectively; Constraint (\ref{1b}) is stricter than Constraint (\ref{2b}); 
the combination of Constraints (\ref{1b}) and (\ref{1c}) is stricter than Constraint (\ref{2d}); finally, Constraint~\eqref{fbQ1} is stricter than Constraint~\eqref{fbQ2inner2}.

\begin{Remark}\label{remark2}
 For  $\tilde{Y}=$ const., both $\set{R}_\textnormal{in,1}$ and $\set{R}_\textnormal{in,2}$ specialize to $\set{C}_\textnormal{NoFB}$. Thus, both regions $\set{R}_\textnormal{in,1}$ and $\set{R}_\textnormal{in,2}$ include $\set{C}_\textnormal{NoFB}$.
\end{Remark}

\begin{Theorem}\label{theo3}
Assume $R_\textnormal{FB}>0$. For strictly less-noisy DMBCs: 
\begin{enumerate}
\item
every rate pair $(R_1>0,R_2>0)$ on the boundary of $\set{C}_{\textnormal{NoFB}}$ that is \emph{not} also on the boundary of $\set{C}_{\textnormal{Enh}}$ can be improved with rate-limited feedback:
\begin{IEEEeqnarray}{rCl}
\lefteqn{
(R_1>0, R_2>0) \in \big( \textnormal{bd}(\set{C}_\textnormal{NoFB}) \cap \textnormal{int}(\set{C}_\textnormal{Enh})
 \big)} \hspace{3cm}\nonumber \\& &
 \Rightarrow  (R_1, R_2)  \in  \textnormal{int}(\set{C}_\textnormal{FB}(R_\textnormal{FB})). \IEEEeqnarraynumspace
\end{IEEEeqnarray}
\item
whenever $\set{C}_{\textnormal{NoFB}}$ does not coincide with $\set{C}_{\textnormal{Enh}}$, then the feedback capacity  region is strictly larger than the no-feedback capacity region:
\begin{equation}
\big( \set{C}_{\textnormal{NoFB}} \subset \set{C}_{\textnormal{Enh}}\big) \Rightarrow \big( \set{C}_{\textnormal{NoFB}} \subset \set{C}_{\textnormal{FB}}(R_{\textnormal{FB}})\big).
\end{equation}
\end{enumerate}
 \end{Theorem}
\begin{proof}
Statement 1) is proved in Section \ref{sec:proofTheorem3}. Statement 2) follows directly from Statement 1).
\end{proof}
(Notice that for physically degraded DMBCs, $\set{C}_{\textnormal{NoFB}} = \set{C}_{\textnormal{Enh}}$ and in fact, feedback does not increase capacity.)
\begin{Corollary}
For BS-BCs with cross-over probabilities $p_1, p_2$ that satisfy $1/2>p_1>p_2>0$, rate-limited feedback increases the entire capacity region irrespective of  $R_{\textnormal{FB}}>0$.

The same statement holds also for BE-BCs with erasure probabilities $\delta_1, \delta_2$ that satisfy $1>\delta_1>\delta_2>0$.
\end{Corollary}

\section{Proof of Theorem \ref{theo1} }\label{sec:sche1}
We first describe a scheme achieving $\set{R}_{\textnormal{in,1}}$ for $|\set{Q}|=1$ (Section~\ref{subsec:sche1}). Due to space limitations, we only sketch the required modifications for $|\set{Q}|\geq 2$ (Section~\ref{subsec:sche2}).
\subsection{Scheme achieving rate region $\set{R}_\textnormal{in,1}$ for $|\set{Q}|=1$}\label{subsec:sche1}

Let $\epsilon>0$. Fix a pmf $P_{U}P_{X|U}P_{Y_1Y_2|X}P_{\tilde{Y}|UY_1}$ and positive rates $R_1$ and $R_2$ such that the constraints in (\ref{rate1}) and~\eqref{fbQ1} hold with strict inequality. Choose a positive rate $\tilde{R}$ that satisfies
 \begin{IEEEeqnarray}{rCl}\label{RateconsFB}
    I(\tilde{Y};Y_1|UY_2)+\epsilon&\leq & \tilde{R}\leq  R_\textnormal{FB}.
  \end{IEEEeqnarray}

Transmission takes place over $B+1$ consecutive blocks, with length $n$ for each block. We denote the input and output sequences in block $b\in\{1,\ldots,B+1\}$ by $X^n_{b}, Y^n_{1,b},Y^n_{2,b}$, respectively. The messages to be sent are in a product form $M_i  = (M_{i,1},\ldots,M_{i,B})$, for $i\in\{1,2\}$, where each $M_{i,b}$ is uniformly distributed over the set $\set{M}^{(n)}_i:=\{1,\ldots,\lfloor 2^{nR_i}\rfloor\}$. Let $\tilde{\set{M}}:=\{1,\ldots,\lfloor2^{n\tilde{R}}\rfloor\}$ and $\set{M}:=\{1,\ldots,\lfloor2^{n\hat{R}}\rfloor\}$ with $\hat{R}:=I(\tilde{Y};Y_1|U)+\epsilon$.

\textit{1) Codebook generation}:
For each block $b$,  randomly and independently generate $2^{n(R_1+\tilde{R})}$ sequences $u_b^n(m_{1,b},l_{b-1})$, for $m_{1,b}\in \set{M}^{(n)}_1$ and $l_{b-1}\in \tilde{\set{M}}$. Each sequence $u_b^n(m_{1,b},l_{b-1})$ is drawn according to the product distribution $\prod_{t=1}^nP_U(u_{b,t})$, where $u_{b,t}$ denotes the $t$-th entry of $u_b^n(m_{1,b},l_{b-1})$.

For each pair $(m_{1,b},l_{b-1})$, randomly and conditionally independently generate $2^{nR_2}$ sequences $x_b^n\big(m_{2,b}|(m_{1,b},l_{b-1})\big)$, for $m_{2,b}\in \set{M}^{(n)}_2$. Each sequence $x_b^n\big(m_{2,b}|(m_{1,b},l_{b-1})\big)$ is drawn according to the product distribution $\prod_{t=1}^nP_{X|U}(x_{b,t}|u_{b,t})$, where $x_{b,t}$ denotes the $t$-th entry of $x_b^n\big(m_{2,b}|(m_{1,b},l_{b-1})\big)$.

For each pair $(m_{1,b},l_{b-1})$, randomly and conditionally independently generate $2^{n\hat{R}}$ sequences $\tilde{y}_b^n\big(m_{b}|(m_{1,b},l_{b-1})\big)$, for $m_{b}\in \set{M}$. Each sequence $\tilde{y}_b^n\big(m_{b}|(m_{1,b},l_{b-1})\big)$ is drawn according to the product distribution $\prod_{t=1}^nP_{\tilde{Y}|U}(\tilde{y}_{b,t}|u_{b,t})$, where $\tilde{y}_{b,t}$ denotes the $t$-th entry of $\tilde{y}_b^n\big(m_{b}|(m_{1,b},l_{b-1})\big)$. Partition $\set{M}$ into $2^{n\tilde{R}}$ equal-size subsets referred as bins $\set{B}(l_b)=\{(l_b-1)2^{n(\hat{R}-\tilde{R})}+1,\ldots,l_b2^{n(\hat{R}-\tilde{R})}\}$.

All codebooks are revealed to  transmitter and  receivers. %

\textit{2) Transmitter}:
To simplify notation, define $(l_0,m_{1,B+1},m_{2,B+1},m_{B+1}):=(1,1,1,1)$.

For each block $b\in\{1,\ldots, B+1\}$ and given that $M_{1,b}=m_{1,b}$ and $M_{2,b}=m_{2,b}$ and that the feedback message in block $b-1$ is $L_{b-1}=l_{b-1}$, 
 the transmitter sends
\begin{equation}\label{encoding}x_b^n(m_{2,b}|m_{1,b},l_{b-1}).\end{equation}
(The generation of the feedback signal $L_{b-1}=l_{b-1}$ sent in block $b-1$ is described shortly.)

\textit{3) Receiver~1}:
In each block $b\in\{1,..,B+1\}$, after observing channel outputs $Y^n_{1,b}=y^n_{1,b}$, Receiver~1 first looks for an index $\hat{m}_{1,b}\in \set{M}^{(n)}_1$  that satisfies
\begin{equation} \label{decodeM1}
\big(u^n_b(\hat{m}_{1,b},{l}_{b-1}), y_{1,b}^n\big)\in \set{T}_\epsilon^n(P_{UY_1}).
\end{equation}
Notice that Receiver~1 knows the feedback message $l_{b-1}$, because it has generated it itself in the previous block~$(b-1)$.

Next, it looks for a compression message ${m}_b\in {\set{M}}$  that satisfies
\begin{equation} \label{compressY1}
\big(u^n_b(\hat{m}_{1,b},{l}_{b-1}),\tilde{y}_b^n({m}_b|\hat{m}_{1,b},{l}_{b-1}), y_{1,b}^n\big)\in \set{T}_\epsilon^n(P_{U\tilde{Y}Y_1}),
\end{equation}
and feeds back the index $l_b$ of the bin containing $m_b$, i.e. it feeds back $l_b$ if $m_b\in \set{B}(l_b)$. Thus, Receiver~1 only sends a feedback signal in the last channel use of each block, otherwise it stays silent. By (\ref{RateconsFB}) and because the message $l_b$ is of rate $\tilde{R}$, our scheme satisfies the average feedback-rate constraint~(\ref{consFB0}).

After decoding block $B+1$, Receiver 1 produces the product message
$\hat{m}_1=(\hat{m}_{1,1},\ldots,\hat{m}_{1,B})$ as its guess.

\textit{4) Receiver~2}:
For each block $b\in\{1,\ldots,B+1\}$, after observing $Y_{2,b}^n=y_{2,b}^n$, Receiver 2 looks for a pair $(\hat{m}'_{1,b},\hat{l}_{b-1})\in \set{M}^{(n)}_1\times \tilde{\set{M}}$ that satisfies
\begin{equation}\label{eq:u2}
\big(u_b^n(\hat{m}'_{1,b},\hat{l}_{b-1}),y^n_{2,b}\big)\in \set{T}_\epsilon^n(P_{UY_2}).
\end{equation}
Then, it looks for the compression message $\hat{m}_{b-1}\in \set{B}\big(\hat{l}_{b-1}\big)$ that satisfies
\begin{IEEEeqnarray}{rCl}
\big(u_{b-1}^n&&(\hat{m}'_{1,b-1},\hat{l}_{b-2}),y^n_{2,b-1}, \nonumber\\
&&\quad\quad \tilde{y}_{b-1}^n(\hat{m}_{b-1}|\hat{m}'_{1,b-1},\hat{l}_{b-2})\big)\in \set{T}_\epsilon^n(P_{U\tilde{Y}Y_2}),\IEEEeqnarraynumspace\label{eq:y21}
\end{IEEEeqnarray}
and finally searches for an index $\hat{m}_{2,b-1}\in \set{M}^{(n)}_2$ that satisfies
\begin{IEEEeqnarray}{rCl}\label{eq:y2}
&&\big(u^n_{b-1}(\hat{m}'_{1,b-1},\hat{l}_{b-2}),x^n_{b-1}(\hat{m}_{2,b-1}|\hat{m}'_{1,b-1},\hat{l}_{b-2}), \nonumber\\&&\hspace{0.6cm}\tilde{y}_{b-1}^n(\hat{m}_{b-1}|\hat{m}'_{1,b-1},\hat{l}_{b-2}), y_{2,b-1}^n\big)\in \set{T}_\epsilon^n(P_{UX\tilde{Y}Y_2}).\IEEEeqnarraynumspace
\end{IEEEeqnarray}
After decoding block $B+1$, Receiver 2 produces the product message
$\hat{m}_2=(\hat{m}_{2,1},\ldots,\hat{m}_{2,B})$ as its guess.

\textit{5) Analysis:} 
Using standard typicality arguments 
one can show that the average probability of error of the scheme, \eqref{errorprob},
(averaged over the random messages, the random channel realization and the random code construction) tends to zero as the blocklength $n$ tends to infinity, whenever
\begin{subequations}\label{method}
\begin{IEEEeqnarray}{rCl}
{R}_1&<&I(U;Y_1)-\delta(\epsilon)\\
\tilde{R}&>&I(\tilde{Y};Y_1|U,Y_2)+\delta(\epsilon)\\
{R}_1+\tilde{R}&<&I(U;Y_2)-\delta(\epsilon)\\
{R}_2&<&I(X;\tilde{Y}Y_2|U)-\delta(\epsilon)
\end{IEEEeqnarray}
\end{subequations}
for some function $\delta(\epsilon)$ that tends to 0 as $\epsilon$ tends to 0. This implies that there also is a deterministic code with probability of error tending to 0 as $n$ tends to infinity whenever constraints~\eqref{method} are satisfied.

Applying the Fourier-Motzkin algorithm to the constraints in~\eqref{RateconsFB} and~\eqref{method} to eliminate the rate $\tilde{R}$, and letting $\epsilon$ tend to 0 and the blocklength $B$ to infinity,  establishes the achievability of the region $\set{R}_{in,1}$ when $|\set{Q}|=1$. (Notice that for a finite $B$ the rates of transmission are $\frac{B}{B+1}R_1$ and $\frac{B}{B+1}R_2$.)

\subsection{Sketch of scheme achieving rate region $\set{R}_\textnormal{in,1}$ for $|\set{Q}|> 1$}\label{subsec:sche2}

Let $\set{Q}=\{1,\ldots, |\set{Q}|\}$. Our scheme consists of the phases $1,\ldots, |\set{Q}|$. In each phase~$q\in\set{Q}$, we apply the scheme from the previous subsection~\ref{subsec:sche1}, but where now, the transmitter can delay the transmission of the compression messages $\{l_b\}$ (or parts of them) to subsequent  phases. More specifically:

 For each phase $q\in\set{Q}$, the transmitter and Receiver~2 each have a first-input first-output (FIFO)  queue. 
At the beginning of each block $b$ of  phase $q$, the transmitter stores  in its FIFO queue a bit representation of the feedback-signal $l_{q,b-1}$ it observed in the previous block. To perform the encoding, it retrieves the first $n_q{R}^{(q)}_{\textnormal{fw}}$ bits from this queue, where $n_q$ denotes the blocklength in phase $q$ and $R^{(q)}_{\textnormal{fw}}$ is a new parameter of the scheme. It then follows the encoding described in Section~\ref{subsec:sche1} but where $l_{b-1}$ in~\eqref{encoding} needs to be replaced by  $k_{q,b-1}$, the index corresponding to the retrieved bits, and in general the subscript $b$ must be replaced by the pair $q,b$.

Receiver~2 decodes $\hat{m}_{1,q,b}'$ and $\hat{k}_{q,b-1}$ as in~\eqref{eq:u2} (but where $l_{b-1}$ is replaced by $k_{q,b-1}$ and generally the subscript $b$ is replaced by $q,b$) and stores the decoded index $\hat{k}_{q,b-1}$ in its  FIFO queue. Once the FIFO queue contains all pieces to reconstruct the feedback message $\hat{l}_{q,b-1}$, Receiver~2 decodes the submessage $m_{2,q,b-1}$ as in~\eqref{eq:y21} and \eqref{eq:y2}. 
Receiver~1 performs the same operations as in the scheme in Section~\ref{subsec:sche1},  but where $l_{b-1}$ in~\eqref{decodeM1} and \eqref{compressY1} needs to be replaced by $k_{q,b-1}$ and the subscript $b$ by  $q,b$.
(Receiver~1 knows $\{k_{q,b-1}\}$ as it can simulate the transmitter's  queue.)


We sketch the main points of the analysis. For each $q\in\set{Q}$, let $P_{U|Q=q}P_{X|UQ=q}P_{Y_1Y_2|X}P_{\tilde{Y}|UY_1Q=q}$ denote the distribution used in the code construction of phase $q$. Also, let $R_1^{(q)}$ and $R_2^{(q)}$ denote the rates of transmission of the messages sent in phase $q$ and  $\tilde{R}^{(q)}$ denote the rate of the compression messages $\{l_{q,b}\}_{b=1}^B$. The blocklength in phase $q$ is chosen as $n_q:=P_Q(q)n$ for some large positive integer $n$, where $P_Q(q)$ is a pmf  over $\set{Q}$. For simplicity of exposition, assume that 
the labeling of the elements in $\set{Q}$ is such that
\begin{IEEEeqnarray}{rCl}
\sum_{1\leq q'\leq q}\tilde{R}^{(q')} P_Q(q') &\geq  &\sum_{1\leq q' \leq q} {R}^{(q')}_{\textnormal{fw}} P_Q(q')\label{eq:consQb}
\end{IEEEeqnarray}
holds for all $q\in\set{Q}$. Inequality~\eqref{eq:consQb} ensures that prior to each block of phase $q$ the transmitter finds $n_q R_{\textnormal{fw}}^{(q)}$ bits in its queue. 

Now,  if 
 \begin{equation}
\sum_{q\in\set{Q}}\tilde{R}^{(q)} P_Q(q) =  \sum_{q\in\set{Q}} {R}^{(q)}_{\textnormal{fw}} P_Q(q)\label{eq:lastq}\IEEEeqnarraynumspace\end{equation}
then, at the end of the last phase~$|\set{Q}|$, Receiver~2 has reconstructed all feedback signals $\{\hat{l}_{q,b}\}$. Assuming that Conditions~\eqref{eq:consQb} and \eqref{eq:lastq} hold, if  also the following Conditions~\eqref{eq:1}--\eqref{eq:consQa} are  satisfied for each $q\in\set{Q}$:
 \begin{IEEEeqnarray}{rCl}
R^{(q)}_{1} & < &I(U;Y_1|Q=q) \label{eq:1} \\
R^{(q)}_{1} +R^{(q)}_{\textnormal{fw}} & < & I(U;Y_2|Q=q) \\ 
R^{(q)}_2 &< & I(X;Y_2 \tilde{Y}_1|U Q=q)\\
\tilde{R}^{(q)} &> & I(\tilde{Y};Y_1|UY_2Q=q), \label{eq:consQa} 
\end{IEEEeqnarray}
then  the probability of error tends to 0 as $N\to \infty$. 
Notice that our scheme satisfies the feedback-rate constraint~\eqref{consFB0}, whenever
\begin{equation}\label{eqffb}
 \sum_{q\in\set{Q}}\tilde{R}^{(q)} \leq  R_{\textnormal{FB}}.
 \end{equation}
From Constraints~\eqref{eq:lastq}--\eqref{eqffb}, we obtain the achievable region in Theorem~\ref{theo1}, for example by applying the Fourier-Motzkin algorithm (where we can relax the equality in~\eqref{eq:lastq} to an $\leq$-inequality).

\section{Proof of Theorem \ref{theo3}}\label{sec:proofTheorem3}
Assume that $R_{\textnormal{FB}}>0$.
 Fix  $(R_1^{(1)}>0,R_2^{(1)}>0)$ such that 
 \begin{equation}
(R_1^{(1)},R_2^{(1)})\in \big(\textnormal{bd}(\set{C}_\textnormal{NoFB})\cap {\textnormal{int}(\set{C}_\textnormal{Enh})}\big).
 \end{equation}
(This implies  $R_1^{(1)}>0$ and $R_2^{(1)}>0$.) Since $(R^{(1)}_1, R^{(1)}_2)\in \textnormal{bd}(\set{C}_\textnormal{NoFB})$, there exists a pmf $P_{U^{(1)}X^{(1)}}$  satisfying
 \begin{subequations}
   \begin{IEEEeqnarray}{rCl}
   R_1^{(1)}&:=&I(U^{(1)};Y_1^{(1)})\\
      R^{(1)}_2&:=&I(X^{(1)};Y_2^{(1)}|U^{(1)})
 \end{IEEEeqnarray}
 \end{subequations}
where $(Y_1^{(1)},Y_2^{(1)})\sim P_{Y_1Y_2|X}$, given $X^{(1)}$. Now, since $(R^{(1)}_1, R^{(1)}_2)\in \textnormal{int}(\set{C}_\textnormal{Enh})$, we can find
$(R^{(2)}_1, R^{(2)}_2) \in \textnormal{bd}(\set{C}_\textnormal{Enh})$ satisfying
\begin{IEEEeqnarray}{rCl}\label{Q2vQ1}
R^{(2)}_1>R_1^{(1)}\quad \textnormal{and} \quad  R^{(2)}_2>R^{(1)}_2.
 \end{IEEEeqnarray}
 Since $(R^{(2)}_1, R^{(2)}_2) \in \textnormal{bd}(\set{C}_\textnormal{Enh})$, there is a pmf $P_{U^{(2)}X^{(2)}}$ satisfying
  \begin{subequations}
   \begin{IEEEeqnarray}{rCl}
   R^{(2)}_1&:=&I(U^{(2)};Y_1^{(2)})\\
      R^{(2)}_2&:=&I(X^{(2)};Y_1^{(2)}Y_2^{(2)}|U^{(2)}),
 \end{IEEEeqnarray}
 \end{subequations}
 where $(Y_1^{(2)},Y_2^{(2)})\sim P_{Y_1Y_2|X}$, given $X^{(2)}$.

Choose now $\gamma \in (0,1)$ to satisfy
 \begin{IEEEeqnarray}{rCl}\label{gamma}
\gamma{H^{(2)}}&\leq& \min\{ {R_\textnormal{FB}}, \gamma I^{(2)} + \bar{\gamma}I^{(1)}\} \label{gammaR} \end{IEEEeqnarray}
with
\begin{IEEEeqnarray*}{rCl}
 H^{(2)} &:=& H(Y_1^{(2)}|Y_2^{(2)}U^{(2)}) \\
 I^{(q)} &:=& I(U^{(q)};Y_2^{(q)})-I(U^{(q)};Y_1^{(q)}),~\text{for}~q=1,2. 
 \end{IEEEeqnarray*}
Such a $\gamma$ exists because $R_{\textnormal{FB}}$, $H^{(2)}$, $I^{(1)}$,  and $I^{(2)}$ are positive. (Here, $I^{(1)}$ and $I^{(2)}$ are positive because the DMBC is strictly less-noisy, see (\ref{LN}).)
Then, introduce an independent  random variable $Q$ with
  \begin{IEEEeqnarray}{rCl}
  P_Q(q)=\left\{
  \begin{array}{ll}
1-\gamma, & q=1\\
  \gamma,  & q=2
  \end{array} \right.
 \end{IEEEeqnarray}
and  define
\begin{subequations}\label{eq:pointout}
\begin{IEEEeqnarray}{rClCl}\label{barRQ}
{R}'_1&:=&(1-\gamma)R^{(1)}_1+\gamma R^{(2)}_1
&= & I(U^{(Q)};Y_1^{(Q)}|Q)\\
{R}'_2&:=&(1-\gamma)R^{(1)}_2+\gamma R^{(2)}_2
         &=&I(X^{(Q)};\tilde{Y}^{(Q)}Y_2^{(Q)}|U^{(Q)}Q)\nonumber \\\end{IEEEeqnarray}
 \end{subequations}
with $\tilde{Y}^{(2)}=Y_1^{(2)}$ and $\tilde{Y}^{(1)}$ being a constant.
By (\ref{Q2vQ1}) and since $0<\gamma<1$, we have
\begin{IEEEeqnarray}{rCl}\label{eq:strict}
{R}'_1>R^{(1)}_1\quad \textnormal{and} \quad {R}'_2>R^{(1)}_2.
\end{IEEEeqnarray}
In the rest of the proof we show that the pair $(R_1', R_2')\in\set{C}_{\textnormal{FB}}(R_{\textnormal{FB}})$. Combined with~\eqref{eq:strict} this establises that $(R_1^{(1)}, R_2^{(1)})$ must lie in the interior of $(\set{C}_\textnormal{FB}(R_\textnormal{FB}))$.

Notice first that the pmf of the tuple $(U^{(Q)}, X^{(Q)}, Y_1^{(Q)}, Y_2^{(Q)}, \tilde{Y}^{(Q)}, Q)$
is of  the form
\begin{equation}\label{fbMC}
P_Q P_{U^{(Q)}|Q} P_{X^{(Q)}|U^{(Q)} Q} P_{Y_1^{(Q)} Y_2^{(Q)}|X^{(Q)}} P_{\tilde{Y}^{(Q)}| U^{(Q)}Y_1^{(Q)}Q}
\end{equation}
where
$P_{Y_1^{(Q)}Y_2^{(Q)}|X^{(Q)}}$ coincides with the channel law $P_{Y_1,Y_2|X}$.  This statement holds because by construction it is satisfied given
$Q=1$ and given $Q=2$, and because $Q$ is independent of $\{ (U^{(q)}, X^{(q)}, Y_1^{(q)}, Y_2^{(q)}, \tilde{Y}^{(q)}) \}_{q=1}^2$.

Since $\tilde{Y}^{(1)}=\textnormal{const.}$ and $\tilde{Y}^{(2)}=Y_1^{(2)}$ and by (\ref{gammaR})
\begin{IEEEeqnarray}{rCl}\label{eq:fbstrict}
\lefteqn{
I(\tilde{Y}^{(Q)};Y^{(Q)}_1|U^{(Q)}Y^{(Q)}_2Q)}\nonumber \\
                   &=& \gamma I(\tilde{Y}^{(Q)};Y^{(Q)}_1|U^{(Q)}Y^{(Q)}_2,Q=2) \nonumber \\
   &=&  \gamma      H^{(2)} \nonumber\\
   &\leq &\min\{ R_\textnormal{FB} , I(U^{(Q)};Y_2^{(Q)}|Q)-I(U^{(Q)};Y_1^{(Q)}|Q)\}.\IEEEeqnarraynumspace
\end{IEEEeqnarray}


The fact that $(R_1', R_2')\in\set{C}_{\textnormal{FB}}(R_{\textnormal{FB}})$ follows now by (\ref{eq:pointout}), (\ref{fbMC}), and (\ref{eq:fbstrict}), and by  Corollary~\ref{rateCoro}.


\section{Extension: Noisy Feedback}
The results obtained in Section~\ref{sec:results} (Theorems~\ref{theo1}, \ref{theo2}, and \ref{theo3}) apply also to the related setup where the feedback link is a noisy channel of capacity $R_{\textnormal{FB}}$. For example, to achieve the rates in Theorem~\ref{theo1}, the transmitter alternates between sending the blocks of two independent instances (with same parameters) of the scheme in Section~\ref{sec:sche1}. During each of these blocks, Receiver~1 sends the feedback message pertaining to the preceeding block (that belongs to the other scheme) using a code that achieves the capacity $R_{\textnormal{FB}}$ of the feedback link. Decoding at the receivers is performed as before. The transmitter now has to decode the compression message sent over the feedback link, which can be erroneous. However, it can be shown that these additional error events do not change the set of achievable rates.

\section*{Acknowledgement}
This work has been supported by the city of Paris under the program ``Emergences".

\end{document}